\documentclass[11pt]{article} 
\usepackage[text={6.5in,9in}]{geometry}
\usepackage{amsmath,latexsym,amssymb,color,amsthm}
\usepackage{ifthen,graphics,epsfig}
\usepackage{color}
\usepackage{xspace}
\usepackage{enumitem}
\usepackage{url}
\usepackage{amssymb}
\usepackage{pifont}
\usepackage{authblk}
\usepackage{booktabs}
\usepackage[colorlinks=true,citecolor=gray]{hyperref}
\usepackage{caption}
\usepackage{subcaption}
\usepackage[table]{xcolor}

\usepackage{float}
\newfloat{algorithm}{thp}{lop}
\floatname{algorithm}{Algorithm}
\newcommand{\Xomit}[1]{}

\usepackage{marginnote}
\usepackage[table]{xcolor}
\usepackage{todonotes}
\usepackage{menukeys}

\newcommand{\remove}[1]{}

\newlength {\squarewidth}

\newtheorem{theorem}{Theorem}
\newtheorem{lemma}{Lemma}

\newcommand{\toto}{xxx}

\newcounter{linecounter}
\newcommand{\linenumbering}{\ifthenelse{\value{linecounter}<10}{(0\arabic{linecounter})}{(\arabic{linecounter})}}
\renewcommand{\line}[1]{\refstepcounter{linecounter}\label{#1}\linenumbering}
\newcommand{\resetline}[1]{\setcounter{linecounter}{0}#1}
\renewcommand{\thelinecounter}{\ifnum \value{linecounter} > 9\else 0\fi \arabic{linecounter}}

\newcommand{\signed}[1]{\theta_{#1}}
\newcommand{\signedm}[2]{\big( #2, \theta_{#1} \big)}
\newcommand{\aux}[2]{\textsc{aux}[#1](#2)}
\newcommand{\auxproof}[1]{\langle#1\rangle}
\newcommand{\proofs}{\mathit{proofs}}
\newcommand{\isvalid}{{\sf is\_valid}}
\newcommand{\auxvalues}{\mathit{aux\_values}}

\newcommand{\coordmod}{(r_i \modulo n)}

\newcommand{\send}{\mathit{\sf send}}
\newcommand{\sto}{\mathit{\sf to}}
\newcommand{\receive}{\mathit{\sf receive}}
\newcommand{\broadcast}{\mathit{\sf broadcast}}

\newcommand{\ttrue}{\mathit{\tt true}}
\newcommand{\tfalse}{\mathit{\tt false}}

\newcommand{\binpropose}{{\sf bin\_propose}}
\newcommand{\performbroadcast}{{\sf perform\_broadcast}}

\newcommand{\wait}{{\sf wait\_until}}
\renewcommand{\return}{{\sf return}}

\newcommand{\BAMP}{{\cal BAMP}_{n,t}}
\newcommand{\ES}{\mathit{\Diamond Synch}}

\newcommand{\values}{\mathit{values}}

\newcommand{\modulo}{{\sf ~mod~}}
\newcommand{\decide}{{\sf decide}}

\title{A Simple and Efficient Binary Byzantine Consensus Algorithm using Cryptography and Partial Synchrony}

\author[1]{Tyler Crain\thanks{This work was partially done while Tyler Crain was a research assistant at The University of Sydney,
 supported by the Australian Research Council's Discovery Projects funding scheme (project number 180104030).}}
\affil[1]{tcrainwork@gmail.com}

\begin{document}

\maketitle

\begin{abstract}
  This paper describes a simple and efficient Binary Byzantine faulty tolerant
  consensus algorithm using a weak round coordinator and the partial synchrony assumption to ensure liveness.
  In the algorithm, non-faulty nodes perform an initial broadcast followed by a executing a series of rounds
  consisting of a single message broadcast until termination.
  Each message is accompanied by a cryptographic proof of its validity.
  In odd rounds the binary value $1$ can be decided, in even round $0$.
  Up to one third of the nodes can be faulty and termination is ensured within a
  number of round of a constant factor of the number of faults.
  Experiments show termination can be reached in less than $200$ milliseconds with
  $300$ Amazon EC2 instances spread across $5$ continents even with partial initial
  disagreement.
\end{abstract}













\section{Introduction and related work.}
Binary byzantine consensus concerns the problem of getting a set of distinct processes
distributed across a network to agree on a single binary value $0$ or $1$ where processes
can fail in arbitrary ways.
It is well known that this problem is impossible in an asynchronous network with at least
one faulty process~\cite{FLP85}. To get around this, algorithms can employ
randomization~\cite{A03, BO83, B87, BT83, CR93, FP90, KS16, MMR14, PCR14, R83, T84},
or rely on an additional synchrony assumption.
This work assumes partial synchrony~\cite{DDS87, DLS88} which ensures that after some point in time
there exists an (unknown) upper bound on message delay and difference in speed between processes.
Furthermore, the algorithm assumes that less than one third of the processes are faulty 
and ensures termination in $O(t)$ messages steps, both of which are well known lower bounds~\cite{LSP82, FL82}.

While there are many algorithms that solve this problem with the same bounds and assumptions~\cite{DDS87, DLS88},
this algorithm focuses on simplicity and efficiency.
Namely, it starts with each process broadcasting an initial proposal, then executing a series
of rounds that consist of broadcasting a single message then waiting to receive
a threshold of valid messages from other processes.
In the good case agreement happens in the first round.
To ensure termination with the
benefit of the partial synchrony assumption, the algorithm uses a {\it weak round coordinator}~\cite{CGLR18}.
This coordinator is
used to help agreement by suggesting a value to decide when both $0$ and $1$ are valid.
While classic round coordinators~\cite{CT96,DLS88} will rely on the coordinator for termination in all cases,
the weak round coordinator is not needed for termination in the expected case.
Note that this is a practical distinction as the bounds for termination are not effected.

The design of the algorithm is primarily based on two previous algorithms;~\cite{CKS05} and ~\cite{CGLR18}.
Similar to~\cite{CKS05}, a set of cryptographic signatures are included with each message proving
its validity. Differently~\cite{CKS05} uses randomization for termination and each process broadcasts
several messages per round.
Similar to~\cite{CGLR18}, the algorithm can decide $1$ in odd rounds or $0$ in even rounds, uses a weak round coordinator, and
assumes partial synchrony for termination.
Differently,~\cite{CGLR18} has the advantage of not requiring
cryptographic signatures, but requires processes to broadcast several messages per round.
Note that the design of~\cite{CGLR18} is based off of the randomized algorithm of~\cite{MMR14}.

While the binary consensus problem only allows process to agree on a single binary value,
there exist many reductions to multi-value consensus~\cite{MR17, MRT00, TC84, ZC09} allowing processes to agree on arbitrary values.
Furthermore many algorithms~\cite{BSA14, CL02} exists that solve multi-value consensus directly with the same assumptions.
Additionally algorithms exists that make many different assumptions about the model
such as synchrony~\cite{FM97}, different fault models~\cite{LVCQV16, MA06, PSL80}
solve different definitions of consensus~\cite{NCV05}, and so on.

\section{A Byzantine Computation Model.}
\label{sec:model}

This section describes the assumed computation model.

\paragraph{Asynchronous processes.}
The system is made up of a set $\Pi$ of $n$ asynchronous sequential processes,
namely $\Pi = \{p_1,\ldots,p_n\}$; $i$ is called the ``index'' of $p_i$. 
``Asynchronous'' means that each process proceeds at its own speed,
which can vary with time and remains unknown to the other processes.
``Sequential'' means that a process executes one step at a time.
This does not prevent it from executing several threads with an appropriate
multiplexing. 
%
Both notations
$i\in Y$ and $p_i\in Y$ are used to say that $p_i$ belongs to the set $Y$.

\paragraph{Communication network.}
\label{sec:basic-comm-operations}
The processes communicate by exchanging messages through
an asynchronous reliable point-to-point network. ``Asynchronous''  means that
there is no bound on message transfer delays, but these delays are finite.
``Reliable'' means that the network does not lose, duplicate, modify, or
create messages. ``Point-to-point'' means that any pair of processes
is connected by a bidirectional channel.
%
A process $p_i$ sends a message to a process $p_j$ by invoking the primitive 
``$\send$ {\sc tag}$(m)$ $\sto~p_j$'', where {\sc tag} is the type
of the message and $m$ its content. To simplify the presentation, it is
assumed that a process can send messages to itself. A process $p_i$ receives 
a message by executing the primitive ``$\receive()$''.
The macro-operation $\broadcast$ {\sc tag}$(m)$ is  used as a shortcut for
``{\bf for each} $p_i \in \Pi$  {\bf do} $\send$ {\sc tag}$(m)$ $\sto~p_j$
{\bf end for}''. 

\paragraph{Signatures.}
Asymmetric cryptography allow processes to sign messages.
Each process $p_i$ has a public key known by everyone and a private key
known only by $p_i$.
All messages are signed using the private key can be validated by any process
with the corresponding public key, allowing the process to identify the signer of the message.
Signatures are assumed to be unforgeable.
A process will ignore any message that is malformed or contains an invalid signature.

\paragraph{Failure model.}
Up to $t$ processes can exhibit a {\it Byzantine} behavior~\cite{PSL80}.
 A Byzantine process is a process that behaves
arbitrarily: it can crash, fail to send or receive messages, send
arbitrary messages, start in an arbitrary state, perform arbitrary state
transitions, etc. Moreover, Byzantine processes can collude 
to ``pollute'' the computation (e.g., by sending  messages with the same 
content, while they should send messages with distinct content if 
they were non-faulty). 
A process that exhibits a Byzantine behavior is called {\it faulty}.
Otherwise, it is {\it non-faulty}.  
Let us notice that, 
as unforgeable signatures are used
no Byzantine process can impersonate another process.   
Byzantine processes can control the network by modifying
the order in which messages are received, but they cannot
postpone forever message receptions.  

\paragraph{Additional synchrony assumption.}
It is well-known that there is no consensus algorithm ensuring both
safety and liveness properties in fully asynchronous message-passing
systems in which even a single process may crash~\cite{FLP85}.  As the
crash failure model is less severe than the Byzantine failure model,
the consensus impossibility remains true if processes may commit
Byzantine failures.
To circumvent such an impossibility, and ensure the consensus termination
property, 
we enrich the model with additional synchrony
assumptions as follows:
%
After some finite time,
there exists an (unknown) upper bound on message transfer delays.
Furthermore there is an (unknown) upper bound on the difference in speeds between non-faulty processes.
This partial synchrony assumption is denoted $\ES$~\cite{DDS87, DLS88}.



\section{Binary Byzantine Consensus.}
\label{sec:byz-consensus}

\subsection{The Binary Consensus Problem.}


%

In this problem processes input a value to the algorithm, called their \emph{proposal},
run an algorithm consisting of several rounds,
and eventually output a value called their \emph{decision}.
Let $\cal V$ be the set of values that can be proposed.
While  $\cal V$ can contain any number ($\geq 2$) of values
in multi-valued consensus, it contains only two values in binary consensus, 
e.g., ${\cal V} =\{0,1\}$.
Assuming that
each non-faulty process proposes a value, the binary Byzantine consensus (BBC) problem is for 
each of them to
decide on a value in such a way that the following properties are
satisfied:
\begin{itemize}
\item BBC-Termination. Every non-faulty process eventually decides on a value.
\item BBC-Agreement.   No two non-faulty processes decide on different values.
\item BBC-Validity.  If all non-faulty processes propose the same value, no
other value can be decided.
\end{itemize}

\paragraph{Notations.}

\begin{itemize}
\item The acronym ${\BAMP}[\emptyset]$ is used to denote the 
  basic Byzantine Asynchronous Message-Passing computation model;
  $\emptyset$ means that there is no additional assumption. 
\item The basic computation model strengthened with the additional constraint $t<n/3$
  is denoted ${\BAMP}[t<n/3]$.
\item The computation model strengthened with the partial synchrony constraint $\ES$
  is denoted ${\BAMP}[t<n/3,\ES]$.
\item A signature of process $i$ is $\signed{i}$
\item A message $m$ signed by process $i$ is $\signedm{i}{m}$
\end{itemize}


\subsection{A Safe and Live Consensus Algorithm in ${\BAMP}[t<n/3,\ES{}]$.}
\label{ssec:live-bbc}

\paragraph{Message types.}
The following message types are used by the consensus.
\begin{itemize}
\item $\aux{r}{v}$. An {\sc aux} message contains a round number $r$ and a binary value $v$.
\item $\auxproof{\signedm{i}{\aux{r}{v}},\proofs}$. A tuple containing an {\sc aux} message signed
  by process $i$ and a set $\proofs$ containing signed
  {\sc aux} messages from a previous round that are used to prove $v$ is a \emph{valid} binary proposal for round $r$.
\end{itemize}

\paragraph{Valid Notation.}
For a given round $r \geq 1$ a binary value $b$ is \emph{valid} if
$b$ has been proposed by a non-faulty process and
$\neg b$ has not been decided in any round before $r$.
An $\auxproof{\signedm{i}{\aux{r}{v}},\proofs}$ is valid if binary value $v$ is valid in round $r$.
The following section describes a function that is used to compute the validity of a message
given $r$, $v$, and $\proofs$ as input.

\paragraph{Variables.}
The following variables are used throughout all rounds of the consensus.
\begin{itemize}
\item $r_i$. Current round number of process $i$.
\item $\auxvalues_i$. Set of valid signed {\sc aux} messages received by process $i$
  throughout all rounds of the consensus.
\item $timers_i$. A map from a round to a timer at process $i$.
  There are two timers per round so the timers for round $r$ are entries $r \times 2$ and $r \times 2 + 1$
  in the map.
  Timer variables can be started then expire after a predetermined amount of time
  (see Section~\ref{sec:timers} for how times are chosen).
  Starting a timer that has already started or expired is a no-op.
\end{itemize}

\begin{figure*}[ht!]
\centering{
\fbox{
\begin{minipage}[t]{150mm}
\footnotesize
\renewcommand{\baselinestretch}{2.5}
\resetline
\begin{tabbing}
aaaA\=aaA\=aaaA\=aaaaaaaaaA\kill

{\bf opera}\={\bf tion} ${\binpropose}(v_i)$ {\bf is}\\

\line{BBC-01} \> $r_i \leftarrow 0$; $\auxvalues_i \leftarrow \emptyset$; \\


\line{BBC-02} \> ${\broadcast}$ $\auxproof{\signedm{i}{\aux{r_i}{v_i}},\emptyset}$;
{\it \scriptsize  \hfill \color{gray}{// Broadcast the initial proposal}}\\


\line{BBC-03} \> {\bf while} $(\ttrue)$  {\bf do}\\

\line{BBC-04} \>\> $r_i \leftarrow r_i+1$;\\

\line{BBC-05} \>\> {\bf if} ($i = \coordmod$) {\bf then} ${\performbroadcast}()$ {\bf end if};
{\it \scriptsize  \hfill \color{gray}{// Coordinator broadcasts before the timer}}\\

\line{BBC-06} \>\> Start $timers_i[r_i \times 2]$ if not yet done; $\wait$ $timers_i[r_i \times 2]$ has expired; \\

\line{BBC-07} \>\> {\bf if} ($i \neq \coordmod$) {\bf then} ${\performbroadcast}()$ {\bf end if};
{\it \scriptsize  \hfill \color{gray}{// Non-coordinators broadcast after the timer}}\\

\line{BBC-08} \>\> $\wait$ $(n-t)$ valid $\aux{r_i}{}$ messages have been received from $(n-t)$ different processes; \\

\line{BBC-09} \>\> Start $timers_i[r_i \times 2 + 1]$ if not yet done; $\wait$  $timers_i[r_i \times 2 + 1]$ has expired;  \\




\line{BBC-10} \>\> $b_i \leftarrow r_i \modulo 2$; \\

\line{BBC-11} \>\> {\bf if} \= ($n-t$) valid $\aux{r_i}{b_i}$ messages have been received from $(n-t)$ different processes \\

\line{BBC-12} \>\>\> {\bf then} $\decide(b_i)$ if not yet done {\bf  end if} \\






\line{BBC-13} \> {\bf end while}. \\~\\

{\bf proce}\={\bf dure} ${\performbroadcast}()$ {\bf is}\\

\line{BBC-14} \> $values_i \leftarrow$ compute $values_i$ as the set of binary values that satisfy the $\isvalid$ predicate for \\
\>\>   $~~~~~~~~~~$ $\auxvalues_i$ and round $r_i$; \\

\line{BBC-15} \> $bv_i \leftarrow (r_i + 1) \modulo 2$; \\

\line{BBC-16} \> {\bf if} received $\auxproof{\signedm{\coordmod}{\aux{r_i}{p\_val}}}$ from process $p_{\coordmod}$
where $p\_val \in \values_i$ \\

\line{BBC-17} \>\> {\bf then} $est_i \leftarrow p\_val$
{\it \scriptsize  \hfill \color{gray}{// Support the coordinator's value for the next round}}\\

\line{BBC-18} \>\> {\bf else if} $(bv_i \in \values_i)$ {\bf then} $est_i \leftarrow bv_i$
{\it \scriptsize  \hfill \color{gray}{// Otherwise prefer the modulo of the next/previous round}}\\

\line{BBC-19} \>\> {\bf else} $est_i \leftarrow \neg bv_i$ {\bf end if}; \\

\line{BBC-20} \> $proofs_i \leftarrow $ compute $proofs_i$ as a set of signed {\sc aux} messages from $\auxvalues_i$ that satisfy the \\

\>\>   $~~~~~~~~~$  $\isvalid$ predicate for binary value $est_i$ and round $r_i$; \\

\line{BBC-21} \> ${\broadcast}$ $\auxproof{\signedm{i}{\aux{r_i}{est_i}},proofs}$. \\~\\

{\bf when} $\auxproof{\signedm{j}{\aux{r_j}{est_j}},proofs}$  {\bf is received}\\

\line{BBC-22} \>   
   {\bf if} \=
   $\big(\isvalid(r_j, est_j, proofs)\big)$   {\bf then} \\

   \line{BBC-23} \> \> $\proofs \leftarrow proofs \setminus $ $\{$any messages in $\proofs$ not needed to satisfy the $\isvalid$ predicate$\}$. \\
   
\line{BBC-24} \> \>  
  $\auxvalues_i \leftarrow \auxvalues_i \cup \{\signedm{j}{\aux{r_j}{est_j}}\} \cup proofs$; \\


  \line{BBC-25} \>  {\bf end if};\\

  \line{BBC-26} \> $\rho_i \leftarrow$ the largest round for which $\auxvalues_i$ contains $t+1$ messages from different processes; \\
  \line{BBC-27} \> For all $v_i \in \mathbb{N}$ such that $v_i \leq \rho_i \times 2$, set $timers_i[v_i]$ to expired if not yet done.
       {\it \scriptsize  \hfill \color{gray}{// Catch-up mechanism}}

\end{tabbing}
\normalsize
\end{minipage}
}
\caption{A safe algorithm for the binary Byzantine consensus in ${\BAMP}[t<n/3]$.}
\label{algo-BBC} 
\vspace{-1em}
}
\end{figure*}

\begin{figure*}[ht!]
\centering{
\fbox{
\begin{minipage}[t]{150mm}
\footnotesize
\renewcommand{\baselinestretch}{2.5}
\resetline
\begin{tabbing}
aaaA\=aaA\=aaaA\=aaaaaaaaaA\kill

{\bf pred}\={\bf icate} ${\isvalid}(r, est, proofs)$ {\bf is}\\

\line{IV-01} \> {\bf if} $(r=0)$   {\bf then} $\return(\ttrue)$ {\bf  end if};\\

\line{IV-02} \> {\bf if} $(r=1)$ {\bf then} \\

\line{IV-03} \>\> {\bf if} ($\exists$ signed messages $\aux{0}{est}$ from $t+1$ different processes in $\proofs$) \\

\>\>\>  {\bf then} $\return(\ttrue)$ \\

\>\>\> {\bf else} $\return(\tfalse)$ \\

\>\> {\bf end if}; \\

\line{IV-04} \> {\bf end if}; \\

\line{IV-05} \> $b \leftarrow (r-1) \modulo 2$;\\

\line{IV-06} \> {\bf if} $(b=est)$ {\bf then}\\

\line{IV-07} \>\> {\bf if} $(r=2$ $\wedge$ $\exists$ signed messages $\aux{0}{b}$ from $t+1$ different processes in $proofs)$ \\

\>\>\>  {\bf then} $\return(\ttrue)$ \\




\line{IV-08} \>\> {\bf else if} $(\exists$ signed messages $\aux{r-2}{b}$ from $n-t$ different processes in $proofs)$ \\

\>\>\> {\bf then} $\return(\ttrue)$; \\

\line{IV-09} \>\> {\bf end if}; \\

\line{IV-10} \> {\bf else} \\

\line{IV-11} \>\> {\bf if} $(\exists$ signed messages $\aux{r-1}{\neg b}$ from $n-t$ different processes in $proofs)$ \\

\>\>\> {\bf then} $\return(\ttrue)$ \\
\>\> {\bf end if}; \\

\line{IV-12} \> {\bf end if}; \\

\line{IV-13} \> $\return(\tfalse)$. \\

\end{tabbing}
\normalsize
\end{minipage}
}
\caption{Algorithm for the $\isvalid$ predicate.}
\label{algo-is-safe} 
\vspace{-1em}
}
\end{figure*}

\paragraph{Algorithm description.}
Figures~\ref{algo-BBC} and~\ref{algo-is-safe} describe the pseudo-code
for the algorithm.
The $\binpropose$ operation of Figure~\ref{algo-BBC} contains the main loop of the algorithm.
The $\performbroadcast$ procedure describes the code used to prepare and broadcast valid signed {\sc aux}
messages for each round.
The lines~\ref{BBC-22}-\ref{BBC-26} handle the reception of signed {\sc aux} messages.
The $\isvalid$ predicate of Figure~\ref{algo-is-safe} describes the procedure used
to check if a binary value is valid for a given round and a set of signed {\sc aux} messages.

To start the consensus, each process $p_i$ calls $\binpropose$ with its initial binary proposal $v_i$ (Figure~\ref{algo-BBC}).
Line~\ref{BBC-01} initializes local variables, then on Line~\ref{BBC-02} the process
sends an a signed {\sc aux} message with round $0$, binary value $v_i$, and an
empty set for $\proofs$ as any round $0$ message is considered to be valid.
The process then repeats the while loop of Lines~\ref{BBC-03}-\ref{BBC-12} for each round.

A round starts by incrementing the processes round counter on Line~\ref{BBC-04}.
If the process $p_i$ is the round coordinator (i.e. $i = \coordmod$) it then invokes
$\performbroadcast$ (line~\ref{BBC-05}) to generate and broadcast a valid signed {\sc aux} message.
Otherwise the process waits for a timeout (line~\ref{BBC-06}) before invoking $\performbroadcast$ (line~\ref{BBC-07}).
If the timer is large enough, non-faulty processes will receive the coordinator's {\sc aux} message
before broadcasting their own {\sc aux} message using the same binary value if valid.

Non-faulty processes then wait until $(n-t)$ valid {\sc aux} messages have been received for the round,
before waiting on a second timer to expire (lines~\ref{BBC-08}-\ref{BBC-09}).
If the timer is large enough it will ensure all non-faulty processes receive valid signed {\sc aux} messages
from all non-faulty processes for that round.
If $n-t$ of the signed {\sc aux} messages support the same binary value $b_i = r_i \modulo 2$ then the process
decides $b_i$.
The process then continues on to the next round.
Given that $t < n/3$, if a non-faulty process receives $n-t$ signed {\sc aux} messages
in round $r_i$ supporting $b_i$, then any set of $n-t$ signed {\sc aux} from round $r_i$ will contain
at least one {\sc aux} message supporting $b_i$.
With this, the $\isvalid$ predicate will ensure that $\neg b_i$ is not valid in any round after $r_i$
at all non-faulty processes.
As a result, in all following rounds, non-faulty processes will only broadcast {\sc aux}
messages supporting $b_i$ and decide in round $r_i$ or a following round.

Lines~\ref{BBC-14}-\ref{BBC-21} describe the $\performbroadcast$ procedure used to create and
broadcast valid signed {\sc aux} messages for each round.
The procedure starts by using the $\isvalid$ predicate to compute the set of valid binary values for
round $r_i$ given the set $\auxvalues_i$ of valid {\sc aux} messages received so far (line~\ref{BBC-14}).
Lines~\ref{BBC-16}-\ref{BBC-19} are used to decide which of these values to broadcast.
First, if a valid {\sc aux} message from the round coordinator was used, then that value is broadcast(lines~\ref{BBC-16}-\ref{BBC-17}).
Second, $(r_i + 1) \modulo 2$ is chosen if it is valid (line~\ref{BBC-18}),
otherwise the only remaining valid value is chosen (line~\ref{BBC-19}).
By prefering the binary value of the coordinator, termination is ensured when a non-faulty process
is chosen as the coordinator and if large enough timeouts are eventually used,
as all non-faulty processes broadcast the same binary value.
Note that the protocol would remain correct if lines~\ref{BBC-18}-\ref{BBC-19} were changed so that the process
simply chooses a random valid value to broadcast, but they are included as they encourage non-faulty processes to support the same
value and reach a decision even without the presence of the coordinator.

Lines~\ref{BBC-22}-\ref{BBC-27} describe what happens when a signed {\sc aux} message and its proofs are received.
If the $\isvalid$ predicate indicates that this message is valid then the signed {\sc aux} message and its
proofs are added to the $\auxvalues_i$ set (lines~\ref{BBC-22}-\ref{BBC-24}).
Line~\ref{BBC-24} ensures that no invalid messages are added to $\auxvalues_i$.
Next $\rho_i$ is computed as the largest round for which the process has received at least $t+1$ valid
signed {\sc aux} messages from different processes (line~\ref{BBC-25}).
The process then sets all timers up to round $\rho_i$ as expired.
The threshold of $t+1$ ensures that at least one non-faulty processes has reached round $\rho_i$,
then setting the timeouts up to this round as expired helps this process catch up to the faster ones.

\paragraph{Is\_valid predicate description.}
Figure~\ref{algo-is-safe} describes the $\isvalid$ predicate that is called by Algorithm~\ref{algo-BBC}
to check if a binary value is valid.  It takes as input a round $r$, a binary value $est$
and a set of signed {\sc aux} messages $\proofs$.
As previously mentioned, the predicate
should return $\ttrue$ if $\proofs$ ensures that
(i) $est$ was proposed by a non-faulty process
and (ii) $\neg est$ has not been
decided by any non-faulty process in any round before $r$.
Otherwise $\tfalse$ should be returned.

For round $0$ the predicate immediately returns $\ttrue$ as any initial proposal is valid (line~\ref{IV-01}).
For round $1$, as no value can be decided in round $0$, the predicate returns $\ttrue$
if $\proofs$ contains at least $t+1$ round $0$ messages with binary value $est$ (line~\ref{IV-03}), i.e.
if (i) is satisfied.

For all other rounds the value $b \leftarrow (r-1) \modulo 2$ is computed.
Note that $b$ is the binary value that could be decided in round $r-1$.
If $(b=est)$ then $\ttrue$ is returned in the following cases.
\begin{enumerate}
\item $r=2$ and $\proofs$ contains at least $t+1$ signed $\aux{0}{est}$ messages.
  If $r=2$ and $b = est$ then $est = 1$ and by~\ref{BBC-11}-\ref{BBC-12} of Figure~\ref{algo-BBC}
  only $1$ could have been decided in rounds $r < 2$.
  Thus, if $\proofs$ contains at least $t+1$ signed $\aux{0}{est}$ then both (i) and (ii) are satisfied
  and $\ttrue$ is returned (line~\ref{IV-07}).
\item $r > 2$ and $\proofs$ contains at least $n-t$ signed $\aux{r-2}{est}$ messages.
  By lines~\ref{BBC-11}-\ref{BBC-12} of Figure~\ref{algo-BBC}, in round $r-2$ only $\neg b$ could have been decided.
  Furthermore, if $n-t$ signed $\aux{r-2}{b}$ are contained in proofs then all sets of $n-t$ signed {\sc aux}
  messages from round $r-2$ must contain at least one $\aux{r-2}{b}$ message, thus $\neg b$ could not have been
  decided by a non-faulty process in round $r-2$. An induction argument can then be used to show that $\neg b$
  could not have been decided in any round before $r-2$ so $\ttrue$ is returned (line~\ref{IV-08}).
  \label{itm:r2p}
\end{enumerate}

Otherwise if $(b=\neg est)$ then $\ttrue$ is returned only if $n-t$ signed $\aux{r-1}{est}$ messages
are received (line~\ref{IV-11}). This is enough to ensure $\neg est$ was not decided in round $r-1$ or before
using the same argument as case~\ref{itm:r2p} above.
If none of these cases are met the $\tfalse$ is returned.

\subsubsection{Timers}
\label{sec:timers}
The timers used on lines~\ref{BBC-06} and~\ref{BBC-09} of Figure~\ref{algo-BBC} are used to ensure processes
eventually execute rounds synchronously given $\ES$.
The duration of the timers must increase by (at least) a constant amount in every constant number of rounds
to ensure that they are eventually large enough to encompass the bound given by $\ES$ and that slow
non-faulty processes can catch up to the faster ones.
The faster the timeout increases, the fewer rounds will be needed to reach synchrony, but the longer process might
have to wait unnecessarily.
Finding this ideal trade-off would require a detailed inspection of the specific case where the algorithm is expected to run.
Fortunately, as this algorithm uses a weak coordinator, in most cases it does not need a coordinator or timeouts to terminate
(in fact the experiments in Section~\ref{sec:exp} never actually needed to use a coordinator).
Given this it is suggested to set the timeouts to $0$ and disable lines~\ref{BBC-16}-\ref{BBC-17}
of the $\performbroadcast$ procedure that are used to support the coordinator for some constant number of rounds.
If the algorithm does not terminate quickly in this case then it is suggested to set the timeout at the
expected average network delay of the system, and increase it by a small constant each following round.

\subsection{Proofs.}

This section shows that the algorithm presented in Figure~\ref{algo-BBC} solves the Binary consensus problem in ${\BAMP}[t<n/3,\ES]$
through a series of lemmas.

\begin{lemma}
  \label{lem:maj}
  For a given round $r$ there can be at most one binary value $b$ for which there exists at least $n-t$
  signed $\aux{r}{b}$ messages from different processes.
\end{lemma}
\begin{proof}
  This follows from the fact that there are at most $t < n/3$ faulty processes and that non-faulty processes
  sign and broadcast at most one {\sc aux} message per round.
\end{proof}

\begin{lemma}
  \label{lem:propose}
  In round $r > 0$ non-faulty processes will only sign and broadcast {\sc aux} messages containing binary values proposed by non-faulty processes.
\end{lemma}
\begin{proof}
  By line~\ref{BBC-13} of Figure~\ref{algo-BBC} a non-faulty process will only broadcast values that satisfy the $\isvalid$ predicate.
  By lines~\ref{IV-03},~\ref{IV-07},~\ref{IV-08},~\ref{IV-11} of the $\isvalid$ predicate, in any round $r > 0$
  a binary value will only satisfy the predicate if the process has received at least $t+1$ signed {\sc aux} messages
  from different processes.
  Given that there are at most $t$ faults and by induction, non-faulty processes will only broadcast values proposed by non-faulty processes.
\end{proof}

\begin{lemma}
  \label{lem:dec}
All non-faulty processes decide the same value.
\end{lemma}
\begin{proof}

  Assume a non-faulty process decides in round $r_x$.
  By line~\ref{BBC-11} the process must have received $n-t$ signed $\aux{r_x}{r_x \modulo 2}$ messages from different processes
  and decided $v_x=(r_x) \modulo 2$.
  Also by line~\ref{BBC-11} for this or a different non-faulty process to decide $\neg v_x$, the process must
  receive $n-t$ signed $\aux{r_y}{\neg v_x}$ messages from different processes in some round $r_y$.
  Furthermore by lines~\ref{BBC-10} and~\ref{BBC-11}, in round a round $r$ only $(r) \modulo 2$ can be decided, thus $r_y \neq r_x$.

  First assume $r_y > r_x$.
  By Lemma~\ref{lem:maj} no process will receive $n-t$ signed $\aux{r_x}{\neg v_x}$ messages from different processes and
  by line~\ref{BBC-14} a non-faulty process will only sign and broadcast a value that satisfies the $\isvalid$ predicate.
  Given that there are less than $n-t$ signed $\aux{r_x}{\neg v_x}$ messages from different processes, $\neg v_x$
  will not be valid in either round $r_x+1$ or $r_x+2$ (lines~\ref{IV-08},~\ref{IV-11} of the $\isvalid$ predicate),
  thus messages supporting $\neg v_x$ will not be added to $\auxvalues_i$ on line~\ref{BBC-24} of Figure~\ref{algo-BBC} for those rounds,
  and will not be broadcast by non-faulty processes (by lines~\ref{BBC-14},~\ref{BBC-20},~\ref{BBC-21}).
  Given no non-faulty process broadcasts $\neg v_X$ in rounds $r_x+1$ or $r_x+2$,
  the $\isvalid$ predicate will ensure $\neg v_x$ will remain invalid in later rounds and will not
  be broadcast by non-faulty processes in round after $r_x+2$
  (note that the case on line~\ref{IV-02} of the $\isvalid$ predicate does not apply as $r_y > 0$,
  and neither does the case on line~\ref{IV-07} because if $r_y=2$ then $\neg v_x \neq 1$).
  Thus no non-faulty process will decide $\neg v_x$ in a round after $r_x$.

  Next assume $r_y < r_x$.
  If a process receives $n-t$ signed $\aux{r_y}{\neg v_x}$ from different processes and decides ${\neg v_x}$ in round $r_y$
  then using the same argument as above, no non-faulty process will receive $n-t$ signed $\aux{r_x}{v_x}$ in any following round
  and will not decide $v_x$.
  Thus by contradiction no process will decide ${\neg v_x}$ in a round prior to $r_x$.
\end{proof}

\begin{lemma}
  \label{lem:enough}
  For any round $r > 0$ all non-faulty processes will (eventually) receive enough valid messages to
  satisfy the $\isvalid$ predicate of Figure~\ref{algo-is-safe} for the round.
\end{lemma}
\begin{proof}
  
  By line~\ref{IV-01} of the $\isvalid$ predicate all signed round $0$ {\sc aux} messages are valid and
  by line~\ref{BBC-02} of Figure~\ref{algo-BBC} all non-faulty processes sign and broadcast a round $0$ {\sc aux} message.
  All non-faulty processes will then receive at least $n-t$ signed round $0$ {\sc aux} messages from different processes.
  Given that $t < n/3$, of these $n-t$ messages at least $t+1$ messages supporting a single binary
  value will be received, satisfying line~\ref{IV-03} of the $\isvalid$ predicate for round $1$.
  All non-faulty processes will then sign and broadcast a valid {\sc aux} message for round $1$ and advance to round $2$.

  In round $2$ non-faulty processes will receive at least $n-t$ signed valid round $1$ {\sc aux} messages from different processes.
  If $n-t$ of these messages are of the form $\aux{1}{0}$, then by line~\ref{IV-11} of Figure~\ref{algo-is-safe} the $\isvalid$ predicate
  is satisfied for round $2$. Otherwise, at least one of the valid signed {\sc aux} messages must be of the form $\aux{1}{1}$.
  By line~\ref{BBC-22} of Figure~\ref{algo-BBC}
  this message must contain proofs generated by the $\isvalid$ predicate supporting binary value $1$ for round $1$.
  This can only happen on line~\ref{IV-03} of Figure~\ref{algo-is-safe} by including $t+1$ messages of the form $\aux{0}{1}$.
  Notice then that by line~\ref{IV-07} these proofs also satisfy the $\isvalid$ predicate for round $2$.
  Thus, all non-faulty processes will then sign and broadcast a valid {\sc aux} message for round $2$ and advance to round $3$. 

  Now assume by induction all non-faulty processes have received enough valid messages to satisfy the $\isvalid$ predicate
  for a round $r-1$. All processes will then sign and broadcast a valid {\sc aux} message on line~\ref{BBC-21} of
  Figure~\ref{algo-BBC} and advance to round $r$. Following this all non-faulty processes will receive at least $n-t$
  valid signed {\sc aux} messages from round $r-1$. If $n-t$ of these messages are of the form $\aux{r-1}{r \modulo 2}$
  then the $\isvalid$ predicate for round $r$ is satisfied by line~\ref{IV-11} of Figure~\ref{algo-is-safe}.

  Otherwise, at least one of the valid signed {\sc aux} messages must be of the form $\aux{r-1}{\neg (r \modulo 2)}$.
  By line~\ref{BBC-22} of Figure~\ref{algo-BBC}
  this message must contain proofs generated by the $\isvalid$ predicate supporting binary value $\neg(r \modulo 2)$ for round $r-1$.
  For this, on line~\ref{IV-05} of Figure~\ref{algo-is-safe} we have $b = ((r-2) \modulo 2)$, or equivalently $b = (r \modulo 2)$
  and $est = \neg(r \modulo 2)$, i.e. $est \neq b$.
  Therefore by line~\ref{IV-11} the proofs for message $\aux{r-1}{\neg (r \modulo 2)}$ must be $n-t$ messages of the form $\aux{r-2}{\neg (r \modulo 2)}$.
  Now consider the $\isvalid$ predicate with input round $r$ and $est = {\neg r \modulo 2}$, or equivalently $est = ((r-1) \modulo 2)$,
  in this case the predicate is satisfied by $n-t$ messages of the form $\aux{r-2}{\neg (r \modulo 2)}$ (line~\ref{IV-08}),
  which is exactly the set of proofs that were included in the $\aux{r-1}{\neg r \modulo 2}$ message, completing the induction proof.
\end{proof}

\begin{lemma}
  \label{lem:rp2}
  Let $r$ be the smallest round in which a non-faulty process decides. All non-faulty processes will decide in either
  round $r$ or $r+2$.
\end{lemma}
\begin{proof}
  Given line~\ref{BBC-11} of Figure~\ref{algo-BBC}, a non-faulty process decides $v=(r) \modulo 2$ in round $r$ after receiving
  $n-t$ signed $\aux{r}{v}$ from different processes.
  By Lemma~\ref{lem:maj} no process will receive $n-t$ signed $\aux{r}{\neg v}$ messages from different processes,
  thus by Lemma~\ref{lem:dec}, ${\neg v}$ will not satisfy the $\isvalid$ predicate in round any round after $r$.
  From this and by Lemma~\ref{lem:enough}, in all rounds after $r$ the $\isvalid$ predicate on line~\ref{BBC-14}
  will return $\{v\}$ and all non-faulty processes will broadcast {\sc aux} messages supporting $v$.
  Thus by line~\ref{BBC-08} of Figure~\ref{algo-BBC} a non-faulty process will wait until it receives $n-t$
  signed $\aux{r+2}{v}$ messages from different processes, and decide on line~\ref{BBC-12}.
\end{proof}

\begin{lemma}
  \label{lem:sync}
  If all non-faulty processes execute synchronous rounds, then termination is ensured within $O(t)$ rounds.
\end{lemma}
\begin{proof}
  All non-faulty processes executing synchronous rounds will receive all messages from all $n-t$ non-faulty in the round before
  preceding to the next round. Now given two consecutive synchronous rounds $r$ and $r+1$, where coordinators $p_{\coordmod}$
  and $p_{(r+1 \modulo n)}$ are non-faulty processes,
  all non-faulty processes will broadcast {\sc aux} messages with the same binary value in rounds
  $r+1$ and $r+2$ (i.e. the value broadcast by the coordinators by line~\ref{BBC-17} of Figure~\ref{algo-BBC})
  and decide in either round $r$ or $r+1$.
  Given that there are at most $t$ faulty processes and $t < n/3$, two
  consecutive rounds that have non-faulty coordinators will be reached after at most $O(t)$ rounds.
\end{proof}

\begin{lemma}
  \label{lem:term}
  Given the $\ES$ assumption, all non-faulty processes eventually execute synchronous rounds.
\end{lemma}
\begin{proof}
  As described in Section~\ref{sec:timers} the timer for any round $r$ is larger than the timer for round $r-1$ by at least some
  constant $\alpha$.
  By the timers on lines~\ref{BBC-06},~\ref{BBC-09} of Figure~\ref{algo-BBC} and that the threshold for skipping a timer (line~\ref{BBC-26}) is $t+1$,
  no non-faulty process will reach round $r$ faster than the sum of all the timeouts (as given by the fastest
  non-faulty process bounded by $\ES$) of the previous rounds. Thus to reach round $r$, a non-faulty process must have waited
  at least $\sum_{j=0}^{j<r} 2 \times j \times \alpha$ units of time for timers to expire (i.e a polynomial number of time units).
  
  Now given that a non-faulty process will only start the timer on line~\ref{BBC-09} for round $r$ once it has received $n-t$ messages
  and that $t < n/3$, the process must have received messages from at least $t+1$ non-faulty processes for round $r$ when it starts the timer.
  Thus all non-faulty processes will receive $t+1$ messages from round $r$ within a constant bound $c$ given by $\ES$
  and by line~\ref{BBC-27} will skip all timeouts until round $r$ and reach the round in a bound given by $c \times r$
  (i.e. in a linear amount of time units).

  As follows, the fastest non-faulty process reaches round $r$ in a polynomial amount of time bounded by $\ES$, and the slowest
  non-faulty process reaches round $r$ in a linear amount of time bounded by $\ES$.
  Given that a polynomial function grows faster than a linear one,
  eventually the slowest non-faulty processes will reach a round $r$ early enough so that all non-faulty processes
  receive messages from all other non-faulty processes (given $\ES$) before any non-faulty process progresses to round $r+1$.
  Furthermore once this threshold is reached it hold for all following rounds (given $\ES$).
\end{proof}

\begin{theorem}
  The algorithm presented in Figure~\ref{algo-BBC} solves the Binary consensus problem in ${\BAMP}[t<n/3,\ES]$.
\end{theorem}
\begin{proof}
  First recall the definition of Binary Byzantine Consensus.
  \begin{itemize}
  \item BBC-Termination. Every non-faulty process eventually decides on a value.
  \item BBC-Agreement.   No two non-faulty processes decide on different values.
  \item BBC-Validity.  If all non-faulty processes propose the same value, no
    other value can be decided.
\end{itemize}
Lemma~\ref{lem:rp2} ensures that if a non-faulty process decides then all
non-faulty processes decide, while Lemmas~\ref{lem:sync} and~\ref{lem:term} ensure all non-faulty processes decide in the presence
$\ES$, thus ensuring BBC-Termination.
BBC-Agreement and BBC-Validity are ensured by Lemmas~\ref{lem:propose} and~\ref{lem:dec} respectively.
\end{proof}

\section{Implementation and experiments.}

\paragraph{Stopping and garbage collection.}
The algorithm shown in Figure~\ref{algo-BBC} continues to execute rounds forever.
To avoid this, if a non-faulty process decides in round $r$ it can simply broadcast a "proof" of decision, containing the $n-t$
messages that allowed it to decide and stop immediately.
Furthermore, the broadcast of this message may be delayed until the process receives a valid message from another process
from round $r+1$, ensuring that if all processes decide in round $r$ then no extra messages will be sent.
Note that, in implementation, a process can not be immediately garbage collected as it needs
to ensure that its messages are reliably delivered
(reliable channels are often implemented through the use of re-transmissions when needed).
Fortunately, in a system that is executing multiple consensus instances, garbage collection of earlier
instances can be easily coordinated in later instances (this is not described here as it depends on
the requirements of the specific system).

\paragraph{Timeouts and coordinators.}
As described in Section~\ref{sec:timers}, the algorithm does not always need to use timeouts or a coordinator
to terminate.
Disabling timeouts and not using the round coordinator until round $10$ was found to be
a good trade-off, allowing the algorithm to terminate quickly, while still ensuring progress.

\paragraph{Cryptographic signatures and validity proofs.}
Like timeouts and coordinators, including proofs of validity with messages is necessary for the
correctness of the algorithm, but are not often needed in the expected case.
In fact in the presence of reliable channels the validity proofs are needed only in the
case of Byzantine faults.
Given this, for efficiency an implementation may choose not to include proofs with messages by default
and instead have processes request proofs from the sender of the message if the recipient
cannot validate the message itself. Notice that this does not effect the correctness of the algorithm,
but only increases the needed synchrony window to include enough time for a non-faulty process
to request and receive missing proofs from other non-faulty processes.

\paragraph{Threshold signatures.}
Another way to implement validity proofs efficiently is through the use of threshold signatures~\cite{CH89, D88, DF90, R98}.
When using threshold signatures, a set of signatures of the same message from different processes can be combined into
a single shared signature. In the case of this algorithm $n-t$ threshold signatures can be used,
allowing the validity proof to be reduced to a single value in most cases.
To enable threshold signatures, a distributed key generation protocol is usually required to be executed before the first
consensus iteration in order to compute the shared keys.

\subsection{Experiments}\label{sec:exp}
The algorithm has been implemented using the Go~\cite{GO} programming language.
Reliable channels are implemented by using message re-transmission.
All received messages are stored to disk in an append only log allowing processes
to recover after a failure.
Signatures are implemented using the ECDSA implementation included
in the Go standard library~\cite{GOECDSA}. Proofs of validity are transmitted on request
of the recipient as described previously.

The experiments were run on Amazon EC2 using from 75 to 300 c5.large
instances (4 GiB of memory, 2 vCPUs, EBS backed storage).
The instances were spread evenly across EC2's 15 regions
in Asia, Australia, Europe, North America, and South America.

In the experiments each node chooses a random initial binary proposal using
a threshold given by the experiment, then run consensus 10 times.
Results are then calculated as the average of the those runs.
The thresholds are chosen as 25, 50, and 75 percent, where for example
25 percent would mean approximately 25 percent of nodes choose 1 and their
initial proposal with the remaining nodes choosing 0.
All nodes are non-faulty.

Figure~\ref{fig:exp} shows the results of the experiment.
Figure~\ref{fig:lat} shows the average latency of executing a single
consensus instance, Figure~\ref{fig:round} shows the average termination
round the consensus instances, Figure~\ref{fig:msgs} shows the average number
of messages sent for a single consensus instance for all nodes.

Given that the consensus can decide $1$ on round $1$ and $0$ on round $2$,
with $75$ percent $1$ proposals termination happens on the first round, and with
$25$ percent $1$ proposals termination happens on the second round.
With $50$ percent $1$ proposals and given the randomization of the experiment, the termination round
varies in each case, but stays below $2$ on average.
A maximum termination round of $8$ was observed.
The latency is related directly to the termination round, with
$195$ milliseconds being the minimum latency and $293$ milliseconds being the maximum latency.
Increasing the number of nodes from $75$ to $300$ had minimal impact on latency.
The number of messages sent is quadratic to the number of nodes multiplied
by the number of rounds.

\begin{figure*}[ht!]
  \begin{subfigure}{.33\textwidth}
    \includegraphics[scale=0.60]{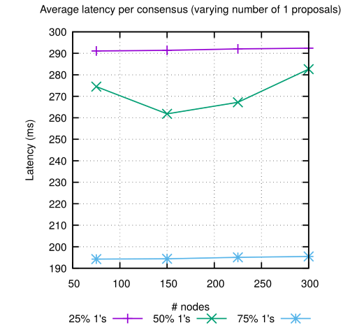}
    \caption{Latency}\label{fig:lat}
  \end{subfigure}
  \begin{subfigure}{.32\textwidth}
    \includegraphics[scale=0.60]{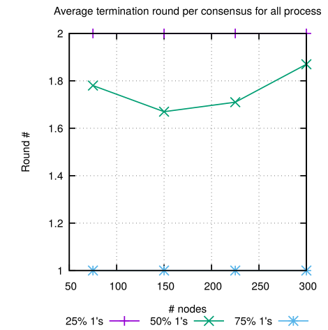}    
    \caption{Termination round}\label{fig:round}
  \end{subfigure}
  \begin{subfigure}{.33\textwidth}
    \includegraphics[scale=0.41]{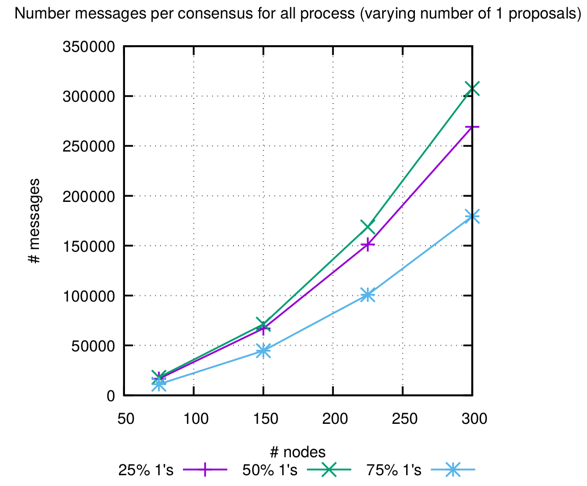}
    \caption{Messages sent}\label{fig:msgs}
  \end{subfigure}
  \caption{Experiment results varying the number of instances and the number of initial 1 proposals.}
  \label{fig:exp} 
\end{figure*}

\end{document}